\documentclass{article}%
\usepackage{amsfonts}
\usepackage{amsmath}
\usepackage{amssymb}
\usepackage{graphicx}%
\newtheorem{theorem}{Theorem}

\newtheorem{proposition}[theorem]{Proposition}

\newenvironment{proof}[1][Proof]{\noindent\textbf{#1.} }{\ \rule{0.5em}{0.5em}}

\def\C{\mathbb{ C}}
\def\R{\mathbb{ R}}
\def\HC{\mathbb{H}(\mathbb{C})}
\def\div{\mathrm{div}}
\def\grad{\mathrm{grad}}
\def\rot{\mathrm{rot}}

\begin{document}
\title{\textbf{On a three-dimensional Riccati differential equation and its symmetries}}
\author{Charles Papillon and S\'{e}bastien Tremblay \medskip\\ D\'epartement de math\'ematiques et d'informatique \\ Universit\'e du Qu\'ebec, Trois-Rivi\`{e}res, Qu\'{e}bec, G9A 5H7, Canada}
\maketitle

\begin{abstract}
\noindent A three-dimensional Riccati differential equation of complex quaternion-valued functions is studied. Many properties similar to those of the ordinary differential Riccati equation such that linearization and Picard theorem are obtained. Lie point symmetries of the quaternionic Riccati equation are calculated as well as the form of the associated three-dimensional potential of the Schr\"odinger equation. Using symmetry reductions and relations between the three-dimensional Riccati and the Schr\"odinger equation, examples are given to obtain solutions of both equations.

\vspace*{2mm}

\noindent PACS 2010: 02.30.Jr, 02.10.Hh, 02.20.Sv

\end{abstract}

\section{Introduction}
The Riccati nonlinear ordinary differential equation \cite{Davis, Ince} is given by
\begin{equation}
\label{riccati1}
y'(x)= p_0(x) + p_1(x) y(x)+ p_2(x)y^2(x), \qquad p_2(x)\neq 0,
\end{equation}
where the prime represents the derivative, plays a very important role in mathematical physics.  This equation is the simplest nonlinear differential equation which being linearizable can be completely solved. Indeed, a classical result for this equation comes from the fact that this equation can be reduced to a second-order linear ordinary differential equation \cite{Ince}
\begin{equation}
\label{linRiccati}
u''-\left(p_1+\frac{p_2'}{p_2}\right)u'+p_0p_2u=0,
\end{equation}
where the solution of this equation will lead to a solution $y=-u'/(p_2u)$ of the Riccati equation (\ref{riccati1}). Conversely by using the Cole-Hopf transformation $y=u'/u$,  for which a generalization is used in this article, the equation
$$
u''+P(x)u'+Q(x)u=0
$$
is transformed into the Riccati equation $y'= -Q(x)-P(x) y-y^2$. The theory of the Riccati equation is therefore equivalent to the theory of the homogeneous linear equation of the second-order. An interesting particular case is the homogeneous one-dimensional Schr\"odinger equation
\begin{equation}
\label{schrodinger1}
-u''+q(x)u=0,
\end{equation}
related to the Riccati equation
\begin{equation}
\label{riccati1p}
y'+y^2=q(x)
\end{equation}
using the Cole-Hopf transformation.

Moreover, the Schr\"odinger operator in \eqref{schrodinger1} can be factorized in the form
\begin{equation}
\label{fact1}
-\frac{d^2}{dx^2}+q(x)=-\left(\frac{d}{dx}+y(x)\right)\left(\frac{d}{dx}-y(x)\right)
\end{equation}
if and only if (\ref{riccati1p}) holds. This observation is the key for the original Darboux transformation \cite{Darboux, Matveev} for the one-dimensional Schr\"odinger equation. The Riccati equation (\ref{riccati1p}) is the equation that we consider for generalization using complex quaternions in this paper.

The general solution of the Riccati equation can be obtained in different ways, depending on the number of particular solutions which are known. From the first Euler theorem's we know that given one particular solution $y_1$ of the Riccati equation (\ref{riccati1}) the general solution $y$ may be obtained using the transformation $u=1/(y-y_1)$. This yields equation (\ref{riccati1}) to the linear first-order nonhomogeneous equation
\begin{equation}
u'+(2y_1p_2+p_1)u+p_2=0
\end{equation}
and a set of solutions to the Riccati equation is then given by $y=y_1+1/u$. Thus given a particular solution of the Riccati equation, it can be linearized and the general solution can be found in two integrations. From the second Euler theorem's we know that giving two particular solutions $y_1,y_2$ of the Riccati equation (\ref{riccati1}) the general solution can be found in one integration, i.e.
\begin{equation}
y=\frac{cy_2 \exp\big[\int p_2(y_1-y_2)dx\big]-y_1}{c\exp\big[\int p_2(y_1-y_2)dx\big]-1},
\end{equation}
where $c$ is a constant. Finally, as was shown by S. Lie \cite{Lie}, the Riccati equation is the only ordinary nonlinear differential equation of first-order which possesses a (nonlinear) superposition formula. For any three solutions $y_1,y_2,y_3$ of the Riccati equation (\ref{riccati1}), the general solution $y$ can be obtained without any integration as
\begin{equation}
y = \frac{y_1(y_3 - y_2) + ky_2(y_1 - y_3) }{y_3 - y_2 + k(y_1 - y_3)},
\end{equation}
where $k$ is a constant. Now from Picard's theorem \cite{Watson}, given a fourth particular solution $y_4$, we have
\begin{equation}
\label{crossratio}
\displaystyle \frac{(y_1-y_2)(y_3-y_4)}{(y_1-y_4)(y_3-y_2)}=k.
\end{equation}
Differentiating \eqref{crossratio} and dividing the result by $(y_1-y_4)(y_3-y_2)(y_1-y_2)(y_3-y_4)$, see \cite{riccati2d}, the Picard's theorem is equivalent to 
\begin{equation}
\label{picardequi}
\frac{(y_1-y_2)'}{y_1-y_2}+\frac{(y_3-y_4)'}{y_3-y_4}-\frac{(y_1-y_4)'}{y_1-y_4}-\frac{(y_3-y_2)'}{y_3-y_2}=0.
\end{equation}

Some generalizations of the Riccati equation have been considered in the literature. In \cite{riccati2d} authors have studied a complex Riccati equation related to the stationnary Schr\"odinger equation. This generalized Riccati equation possesses many interesting properties similar to its one-dimensional prototype and appears for instance in the two-dimensional SUSY quantum mechanics \cite{BilodeauTremblay}. Moreover, up to change of signs, the spatial Riccati equation that we consider in this paper was studied in \cite{Bernstein2, kravchenko, riccati3d} where many interesting results was obtained. In this paper additional results similar to the properties of the (classical) one-dimensional Riccati equation are found: among others generalizations of the first Euler theorem's as well as the Picard's theorem are obtained. Moreover, the Lie point symmetries of the generalized equation are found and used to obtain solutions of the Riccati and the Schr\"odinger equation in the space.

For the reasons mentioned above, Riccati equation (\ref{riccati1}) appears naturally in quantum physics, but this equation also appears in statistical thermodynamics as well as in cosmology \cite{Henkel,Schuch, Wald}. In the fields of applied mathematics, this equation appears in many instances when we can find exact solutions of nonlinear partial differential equations.
Indeed, it plays a very important role in the solution of integrable nonlinear partial differential equations. These equations are characterized by being the compatibility conditions between two linear partial differential equations (the Lax pair) for an auxiliary function, the so called wave function \cite{Ablowitz, Calogero}. It is also the only first-order nonlinear ordinary differential equation which possesses the Painlev\'e property \cite{Ince}, i.e. which has no movable singularity.

\section{Complex quaternions and operators}
The algebra of complex quaternions, also called biquaternions, denoted by $\HC$ consists of elements of the form
\begin{equation}
x=\sum_{\ell=0}^3 x_\ell e_\ell,
\end{equation}
for $x_\ell\in \mathbb{C}$ and the basis elements $e_\ell$ satisfying the following commutativity and multiplication rules
\begin{equation}
e_0e_\ell=e_\ell e_0, \quad e_pe_q=-\delta_{pq}+\varepsilon_{pqr}e_r,\quad ie_\ell=e_\ell i,\quad \ell=0,1,2,3,\quad p,q,r=1,2,3
\end{equation}
where $\delta_{pq}$ is the Kronecker delta,  $\varepsilon_{pqr}$ is the Levi-Civita symbol (i.e. $\varepsilon_{pqr}$ is $1$ if $(pqr)$ is an even permutation or $(123)$, $-1$ if it is an odd permutation, $0$ if any index is repeated) and $i$ is the complex imaginary element. It is useful to represent an element $x\in \HC$ as being the sum of a scalar and a vector part, i.e.
\begin{equation}
x=x_0+\mathbf x, \qquad \mathrm{where} \qquad x_0:=\mathrm{Sc}\ x \quad \mathrm{and}\quad \mathbf x:=  \mathrm{Vec}\ x=\sum_{j=1}^3 x_je_j. 
\end{equation}
The conjugate of an element $x=x_0+\mathbf x\in \HC$ is defined by $\overline x:=x_0-\mathbf x$ and the biquaternionic modulus is given by $|x|^2:=x\overline x=\sum_{\ell=0}^3 x_\ell^2\in \mathbb C$. 

The multiplication of two elements $x,y\in \HC$ can be expressed in the form
\begin{equation}
xy=(x_0+\mathbf x)(y_0+\mathbf y)=x_0y_0+x_0\mathbf y+y_0\mathbf x-\langle \mathbf x,\mathbf y\rangle + \mathbf x \times \mathbf y,
\end{equation}
where $\langle \mathbf x,\mathbf y\rangle$ and $\mathbf x \times \mathbf y$ represent the standard inner and cross products in $\C^3$, respectively.

By $M^x$ we denote the operator of multiplication by a complex quaternion $x$ from the right-hand side, i.e.
\begin{equation}
M^xy=yx, \qquad \text{for all }y\in \HC.
\end{equation}

In what follows and for the rest of the paper we suppose that $\Omega$ is a subdomain of $\R^3$. The Dirac operator $D$, also called the Moisil-Theodoresco operator, is defined as  \cite{gurlebeck}
\begin{equation}
D\varphi:=\sum_{k=1}^3 e_k\,\partial_k \varphi,
\end{equation}
where $\partial_k:= \displaystyle \frac{\partial}{\partial x_k}$ and $\varphi$ is a complex quaternionic valued-function $\varphi\in C^1(\Omega, \HC$). By a straightforward calculation, we find
\begin{equation}
D\varphi=-\div \,\boldsymbol \varphi+\grad\, \varphi_0+\rot\, \boldsymbol \varphi.
\end{equation}
A right Dirac operator $D_r$ can also be defined as $D_r\varphi:=\sum_{k=1}^3 \partial_k \varphi \,e_k$ and is represented by $D_r\varphi=-\div\, \boldsymbol \varphi+\grad \,\varphi_0-\rot \,\boldsymbol \varphi$. Defining the quaternionic conjugate operator $C_H$ as $C_H x=\overline x$ for all $x\in \HC$, the left and the right Dirac operators are related by $C_HD=-D_rC_H$. Moreover, by direct calculations we find that
\begin{equation}
D^2\varphi=-\Delta \varphi
\end{equation}
and obtain the following quaternionic Leibniz rule \cite{gurlebeck2}
\begin{equation}
D[\varphi \, \psi] = D[\varphi]\psi + \overline{\varphi} D[\psi] -2\sum_{k=1}^3 \varphi_k \,\partial_k \psi
\end{equation}
for $\varphi,\psi\in C^1(\Omega,\HC)$. In particular, we observe that if $\mathrm{Vec}\, \varphi=\mathbf 0$, i.e. $\varphi=\varphi_0$, then $D[\varphi_0 \, \psi] = D[\varphi_0]\psi + \varphi_0 D[\psi]$.

Given a known function $\boldsymbol \psi=\sum_{k=1}^3 \psi_k\, e_k \in C^1(\Omega,\mathrm{Vec}\,\HC)$, we consider the linear partial differential equation
\begin{equation}
\grad\,\varphi=\boldsymbol \psi,
\end{equation}
where $\rot \,\boldsymbol \psi=\mathbf 0$ (the compatibility condition) and the unknown function $\varphi$ is in $C^1(\Omega,\C)$. The complex-valued scalar function $\varphi$ is then said to be the potential of the vectorial function $\boldsymbol \psi$. It is well known that we can reconstruct $\varphi$, up to an arbitrary complex constant $C$, in the following way
\begin{equation}
\label{operatorA}
\mathcal A[\boldsymbol \psi](x,y,z)=\int_{x_0}^x \psi_1(\xi,y_0,z_0)d\xi+\int_{y_0}^y \psi_2(x,\eta,z_0)d\eta+\int_{z_0}^z\psi_3(x,y,\zeta)d\zeta+C,
\end{equation}
where $(x_0,y_0,z_0)$ is an arbitrary point in the domain of interest $\Omega$ \cite{Krav Tremb}.

\section{A complex quaternionic Riccati equation}
Let us consider the homogeneous Schr\"odinger equation
\begin{equation}
\label{schrodinger3}
\big(-\Delta +q(x,y,z)\big)\psi=0,
\end{equation}
where $\Delta$ is the three-dimensional laplacian, $q,\psi$ are complex valued-functions such that $q\in C(\Omega,\C)$ and $\psi\in C^2(\Omega,\C)$. Let us also consider the following nonlinear partial differential equation
\begin{equation}
\label{riccati3}
D\mathbf Q+|\mathbf Q|^2=q(x,y,z),
\end{equation}
where $\mathbf Q$ is a vectorial complex quaternionic functions in $C^1(\Omega,\mathrm{Vec}\ \HC)$ and $|\mathbf Q|^2=-\mathbf Q^2=\sum_{j=1}^3 Q_j^2$. We call equation \eqref{riccati3} the {\em three-dimensional Riccati equation} or the {\em (bi)quaternionic Riccati equation} for $\mathbf Q$ in the (bi)quaternions and $q$ a (complex)real-valued function. Decomposing this equation into scalar and vector parts, we find
\begin{equation}
\label{riccatidec}
\begin{array}{rcl}
-\div \, \mathbf Q+|\mathbf Q|^2 &=&q,  \\*[2ex]
\rot \, \mathbf Q &=&\mathbf{0}.
\end{array}
\end{equation}

In \cite{riccati3d} authors have shown a natural counterpart of the relation mentioned in the introduction, between the one-dimensional Schr\"odinger equation (\ref{schrodinger1}) and the Riccati equation (\ref{riccati1p}):

\begin{theorem}
\label{cole3}
\em{\cite{riccati3d}}
The complex-valued function $\psi$ is a solution of the Schr\"odinger equation (\ref{schrodinger3}) if and only if $\mathbf Q=-\displaystyle \frac{D\psi}{\psi}$ is a solution of the Riccati equation (\ref{riccati3}).
\end{theorem}

\begin{theorem}
\label{expA}
The complex quaternionic function $\mathbf Q$ is a solution of the Riccati equation (\ref{riccati3}) if and only if $\psi=\exp\left(-\mathcal A[\mathbf Q]\right)$ is a solution of the Schr\"odinger equation (\ref{schrodinger3}), where $\mathcal A[\mathbf Q]$ is the scalar function defined by (\ref{operatorA}).
\end{theorem}
\begin{proof}
For a given solution $\mathbf Q$ of  (\ref{riccati3}), we have
$$
\begin{array}{rcl}
D^2e^{-\mathcal A[\mathbf Q]} &=& -D\Big(D(\mathcal A[\mathbf Q])e^{-\mathcal A[\mathbf Q]}\Big)=-D\Big(\mathbf Q e^{-\mathcal A[\mathbf Q]}\Big)\\*[2ex]
&=& e^{-\mathcal A[\mathbf Q]}\Big(-D\mathbf Q + \mathbf Q^2\Big)=-e^{-\mathcal A[\mathbf Q]}\Big(D\mathbf Q+|\mathbf Q|^2\Big).
\end{array}
$$
Hence, since $(-\Delta+q)e^{-\mathcal A[\mathbf Q]}  = (D^2+q)e^{-\mathcal A[\mathbf Q]}$ and $q=D\mathbf Q+|\mathbf Q|^2$ we find the desired result.
\end{proof}

The last two theorems represent direct generalization of the relation between the one-dimensional Schr\"odinger equation \eqref{schrodinger1} and Riccati equation \eqref{riccati1p} using the Cole-Hopf transformation.

A factorization of the Schr\"odinger operator can be written in the following form \cite{Bernstein, BernsteinGurl}
\begin{equation}
\label{factorization3}
\begin{array}{rcl}
\big(-\Delta +q(x,y,z)\big)\psi &=& \big(D+M^{\mathbf Q}\big)\big(D-\mathbf Q\,C_{H}\big)\psi \\*[2ex]
&=& \big(D_r+\mathbf{Q} \big) \big(D_r-M^{\mathbf Q}\,C_{H} \big)\psi
\end{array}
\end{equation}
if and only if $\mathbf Q$ is a solution of the Riccati equation (\ref{riccati3}). Here the operators $C_H$, which do not act on the (scalar) complex function $\psi$, are inserted for convenience in what follows. The second factorization obtained in \eqref{factorization3} consists of the application of the operator $C_H$ on the first factorization.

This last result generalizing factorization (\ref{fact1}), combining with theorems \ref{cole3} and \ref{expA}, suggest that (\ref{riccati3}) can be considered a good generalisation of (\ref{riccati1p}).

Considering now the first factorization in \eqref{factorization3} given in terms of two linear first-order partial differential operators, we obtain the following two propositions.
\begin{proposition}
\label{prop1}
{\em \cite{Krav Tremb}}
Let $\varphi$ be a non-vanishing solution of the Schr\"odinger equation \eqref{schrodinger3}. If $W=W_0+\mathbf W\in C^1(\Omega,\HC)$ is a solution of the complex quaternionic equation
\begin{equation}
\label{vekua}
\Big(D-\frac{D\varphi}{\varphi}C_H\Big)W=0,
\end{equation}
then $W_0$ is a solution of the same Schr\"odinger equation \eqref{schrodinger3}. Moreover, the components of $W$ satisfy the equations
\begin{equation}
\label{components1}
\div\left[\varphi^2 \grad\left(\frac{W_0}{\varphi}\right)\right]=0 
\end{equation}
and
\begin{equation}
\label{components2}
\rot\left[\varphi^{-2}\rot\left(\varphi \mathbf W\right)\right]=\mathbf 0.
\end{equation}
\end{proposition} 

Using theorem~\ref{cole3}, an immediate consequence of this last proposition is that $\mathbf Q=-\frac{DW_0}{W_0}$ is solution of the Riccati equation \eqref{riccati3} for a nonvanishing complex-valued function $W_0$ in $\Omega$.

For a given purely vectorial function $\mathbf F(\mathbf x)$ defined for all $\mathbf x=(x_1,x_2,x_3)\in\Omega$, we define $\mathbf B[\mathbf F](\mathbf x)$ by
$$
\mathbf B[\mathbf F](\mathbf x):=\displaystyle \frac{1}{4\pi}\int_{\Omega}\displaystyle \frac{\mathbf F(\mathbf y)}{|\mathbf x-\mathbf y|}d\Omega.
$$
\begin{proposition}
Let $\varphi$ be a non-vanishing solution of the Schr\"odinger equation \eqref{schrodinger3}. If $\mathbf w \in C^1(\Omega,\mathrm{Vec}\,\HC)$ is a purely vectorial solution of the complex quaternionic equation
\begin{equation}
\Big(D+M^{\frac{D\varphi}{\varphi}}\Big)\mathbf w=0,
\end{equation}
then a solution of equation \eqref{vekua} is given by
\begin{equation}
\label{solW}
W=\displaystyle \frac{1}{2}\left(\varphi \mathcal A\left[\frac{\mathbf w}{\varphi}\right]-\varphi^{-1}\rot\big(\mathbf B[\varphi \mathbf w]\big)+\frac{\grad \,h}{\varphi}\right),
\end{equation}
where $h$ is an arbitrary harmonic function in $\Omega$. Moreover, a solution of the Schr\"odinger equation \eqref{schrodinger3} is given by $W_0=\displaystyle \frac{1}{2}\varphi \mathcal A\left[\frac{\mathbf w}{\varphi}\right]$ and a solution of the Riccati equation \eqref{riccati3} is given by 
\begin{equation}
\label{soluQ}
\mathbf Q=-\varphi^{-1}\Big(\grad\,\varphi+\frac{\mathbf w}{\mathcal A\left[\frac{\mathbf w}{\varphi}\right]}\Big).
\end{equation}
\end{proposition}
\begin{proof}
The solution $W$ given by equation \eqref{solW} was obtained in \cite{Krav Tremb} (see theorem 9) and, using theorem~\ref{cole3}, the solution $W_0$ of the Schr\"dinger equation \eqref{schrodinger3} is the scalar part of $W$. Finally for a non-vanishing scalar function $\mathcal A\left[\frac{\mathbf w}{\varphi}\right]$, using theorem  \ref{expA}, we find
$$
\mathbf Q=-\frac{D\left(\varphi \mathcal A\left[\frac{\mathbf w}{\varphi}\right]\right)}{\varphi \mathcal A\left[\frac{\mathbf w}{\varphi}\right]}=-\frac{\mathcal A\left[\frac{\mathbf w}{\varphi}\right]\grad\,\varphi +\mathbf w}{\varphi \mathcal A\left[\frac{\mathbf w}{\varphi}\right]}
$$
which is the desired result \eqref{soluQ}.
\end{proof}

The following proposition will be usefull in the demonstration of the next theorem.
\begin{proposition}
\label{propfortheo}
{\em \cite{Krav Tremb}}
Let $W_0$ be a scalar solution of the Schr\"odinger equation \eqref{schrodinger3} with $q=\frac{\Delta \varphi}{\varphi}$ in $\Omega$. Then the vector function $\mathbf W$ satisfying \eqref{components2}, such that $W=W_0+\mathbf W$ is a solution of \eqref{vekua}, is constructed according to the formula
\begin{equation}
\mathbf W=-\varphi^{-1}\left\{\rot\left(\mathbf B\left[\varphi^2\grad\left(\frac{W_0}{\varphi}\right)\right]\right)+\grad\,h\right\},
\end{equation}
where $h$ is an arbitrary harmonic function.

Given a solution $\mathbf W$ of \eqref{components2}, the corresponding solution $W_0$ of the Schr\"odinger equation \eqref{schrodinger3} with $q=\frac{\Delta \varphi}{\varphi}$ such that $W=W_0+\mathbf W$ is a solution of \eqref{vekua} is constructed as follows
\begin{equation}
W_0=-\varphi \mathcal A\left[\varphi^{-2}\rot(\varphi \mathbf W)\right].
\end{equation}
\end{proposition}

We are now able to consider generalizations of the first Euler's theorem as well as the Picard's theorem (or its equivalent form \eqref{picardequi}).

\begin{theorem}
\label{euler1}
Let $\mathbf Q_1$ be a bounded particular solution of the Riccati equation \eqref{riccati3}. Then the Riccati equation \eqref{riccati3} is reduced to the following first-order equation
\begin{equation}
\label{1order}
DW=-\mathbf Q_1 \overline W
\end{equation}
in the following sense. Any solution of the Riccati equation \eqref{riccati3} has the form
\begin{equation}
\label{dscw}
\mathbf Q=-\displaystyle \frac{D(\mathrm{Sc\,}W)}{\mathrm{Sc\,}W}
\end{equation}
and vice-versa; any solution of equation \eqref{1order} can be expressed via a corresponding solution $\mathbf Q$ of \eqref{riccati3} as follows
\begin{equation}
\label{wconst}
W=e^{-\mathcal A[\mathbf Q]}-e^{\mathcal A[\mathbf Q_1]}\left\{\rot \left(\mathbf B\left[e^{-2\mathcal A[\mathbf Q_1]}\grad (e^{-\mathcal A[\mathbf Q-\mathbf Q_1]})\right]\right)+ \grad\, h\right\},
\end{equation}
where $h$ is an arbitrary harmonic function.
\end{theorem}
\begin{proof}
Let $\mathbf Q_1$ be a bounded solution of \eqref{riccati3}. By using theorem \ref{expA} we have that $\frac{D\varphi}{\varphi}=-\mathbf Q_1$, where $\varphi$ is a nonvanishing solution of the Schr\"odinger equation \eqref{schrodinger3} with $q=\frac{\Delta \varphi}{\varphi}$. Now equation \eqref{1order} can be rewritten as \eqref{vekua}. In the same way we obtain  $\mathbf Q=-\frac{D\psi}{\psi}$, where $\psi$ is a solution of the Schr\"odinger equation \eqref{schrodinger3}. Using proposition \ref{propfortheo}, the function $\psi$ is the scalar part of a solution $W$ of \eqref{vekua}, i.e. $\psi=\mathrm{Sc}\, W$. The first part of the theorem is therefore shown.

Let us suppose now that $W=W_0+\mathbf W$ is a solution of \eqref{1order}. Again, we have that $\frac{D\varphi}{\varphi}=-\mathbf Q_1$. According to proposition \ref{prop1}, $W_0$ is a solution of the Schr\"odinger equation \eqref{schrodinger3} with $q=\frac{\Delta \varphi}{\varphi}$ and $\mathbf Q=-\frac{DW_0}{W_0}$ is a solution of the Riccati equation \eqref{riccati3} by theorem \ref{cole3}.  Using proposition \ref{propfortheo} we obtain 
\begin{equation}
W=W_0-\varphi^{-1}\left\{\rot\left(\mathbf B\left[\varphi^2\grad\left(\frac{W_0}{\varphi}\right)\right]\right)+\grad\,h\right\}. 
\end{equation}
From theorem \ref{expA}, it is now sufficient to make the substitutions $\varphi=e^{-\mathcal A[\mathbf Q_1]}$ and $W_0=e^{-\mathcal A[\mathbf Q]}$ to obtain equation \eqref{wconst}.
\end{proof}

\begin{theorem}
\label{picard}
Let $\mathbf Q_k$, $k=1,2,3,4$, be four solutions of the Riccati equation \eqref{riccati3}. Then we have
\begin{equation}
\label{eqpicard}
\begin{array}{l}
\displaystyle \frac{D(\mathbf Q_1-\mathbf Q_2)-2 (\mathbf Q_1\times \mathbf Q_2)}{\mathbf Q_1-\mathbf Q_2}+\displaystyle \frac{D(\mathbf Q_3-\mathbf Q_4)-2 (\mathbf Q_3\times \mathbf Q_4)}{\mathbf Q_3-\mathbf Q_4} \\*[4ex]
-\displaystyle \frac{D(\mathbf Q_1-\mathbf Q_4)-2 (\mathbf Q_1\times \mathbf Q_4)}{\mathbf Q_1-\mathbf Q_4}-\frac{D(\mathbf Q_3-\mathbf Q_2)-2 (\mathbf Q_3\times \mathbf Q_2)}{\mathbf Q_3-\mathbf Q_2}=\mathbf 0.
\end{array}
\end{equation}
\end{theorem}
\begin{proof}
From the trivial equation
\begin{equation}
C_H\Big[(\mathbf Q_1+\mathbf Q_4)+(\mathbf Q_3+\mathbf Q_2)-(\mathbf Q_1+\mathbf Q_2)-(\mathbf Q_3+\mathbf Q_4)\Big]=\mathbf 0,
\end{equation}
let us consider the first term:
$$
\begin{array}{rcl}
C_H(\mathbf Q_1+\mathbf Q_4)&=& -\displaystyle \frac{(\mathbf Q_1+\mathbf Q_4)(\mathbf Q_1-\mathbf Q_4)}{\mathbf Q_1-\mathbf Q_4} \\*[2ex]
&=& \displaystyle \frac{|\mathbf Q_1|^2+\mathbf Q_1\mathbf Q_4-\mathbf Q_4\mathbf Q_1-|\mathbf Q_4|^2}{\mathbf Q_1-\mathbf Q_4}\\*[2ex]
&=& \displaystyle \frac{(q-|\mathbf Q_4|^2)-(q-|\mathbf Q_1|^2)+\mathbf Q_1\mathbf Q_4-\mathbf Q_4\mathbf Q_1}{\mathbf Q_1-\mathbf Q_4}\\*[2ex]
&=& \displaystyle \frac{D\mathbf Q_4-D\mathbf Q_1+\mathbf Q_1\mathbf Q_4-\mathbf Q_4\mathbf Q_1}{\mathbf Q_1-\mathbf Q_4}\\*[2ex]
&=& -\displaystyle \frac{D(\mathbf Q_1-\mathbf Q_4)-2(\mathbf Q_1\times \mathbf Q_4)}{\mathbf Q_1-\mathbf Q_4}.
\end{array}
$$
Doing similar calculus for the other terms, we obtain the desired result.
\end{proof}

In this last proof, we see that the cross product terms in \eqref{eqpicard} are a consequence of the noncommutativity of (complex) quaternions. The analogue part in the classical Picard's theorem are obviously null.

\section{Symmetries of the spatial Riccati equation}
\subsection{General theory for PDE's}

The purpose to find Lie point symmetries of a system of PDE's is to obtain a local group of transformations that maps every solution of the same system. In other words, it map the solution set of the system to itself. 

Let us consider a general system of differential equations
\begin{equation}
\begin{array}{c}
\label{eqgen}
E^i(x,u,u^{(1)},\ldots,u^{(n)})=0, \\*[2ex]
x\in \R^p,\ u\in \R^q,\ i=1,\ldots, m,\qquad p,q,m,n\in \mathbb Z^{>0},
\end{array}
\end{equation}
where $x$ and $u$ are independent and dependent variables, respectively, and $u^{(k)}$ denotes all partial derivatives of order $k$ of all components of $u$. We are now looking to obtain the local Lie group $\mathcal G$ of local Lie point transformations 
\begin{equation}
\widetilde x=\Lambda_\lambda(x,u),\qquad \widetilde u=\Omega_\lambda(x,u),
\end{equation}
taking solutions of \eqref{eqgen} into solutions, where $\lambda$ are the group parameters. The functions $\Lambda_\lambda$, $\Omega_\lambda$ are well defined for $\lambda$ close to the identity, $\lambda=e$, and for $x$ and $u$ close to some origin in the space $M\subset X\times U$ ($X\sim \R^p$, $U\sim\R^q$) of independent and dependent variables. In other words, the group $\mathcal G$ acts on a manifold $M$ such a manner that whenever $u=f(x)$ is a solution of \eqref{eqgen}, then so is $\widetilde u(\widetilde x)=\lambda\cdot f(x)$. 

Following S. Lie, as presented e.g. in \cite{Olver}, we shall look for infinitesimal group transformations, i.e.  construct the Lie algebra $\mathcal L$ of the Lie group $\mathcal G$. For the purpose of the paper, we shall present the algorithm in a purely utilitarian manner. For the proofs we refer to \cite{Olver}. The Lie algebra $\mathcal L$ will be realized in terms of the vector fields on $M$, i.e. differential operators
\begin{equation}
\label{vhat}
\widehat v=\displaystyle \sum_{i=1}^p \xi^i(x,u)\partial_{x_i}+\sum_{\alpha=1}^q \phi_\alpha(x,u)\partial_{u^\alpha},
\end{equation}
where $\partial_{x_i}:=\frac{\partial}{\partial x_i}$ and $\partial_{u^\alpha}:=\frac{\partial}{\partial u^\alpha}$. The functions $\xi^i$ and $\phi_\alpha$ are to be determined. The algorithm for determining the Lie algebra $\mathcal L$ of the symmetry group $\mathcal G$ of the system \eqref{eqgen} is given by
\begin{equation}
\label{algo}
\mathrm{pr}^{(n)}\widehat v\cdot E^i\Big|_{E^j=0}=0,\qquad i,j=1,\ldots,m,
\end{equation}
where the $n$-th prolongation of the vector fields \eqref{vhat} is given by
\begin{equation}
\begin{array}{c}
\label{prv}
\mathrm{pr}^{(n)}\widehat v=\widehat v+\displaystyle \sum_{\alpha=1}^q \sum_J \phi_{\alpha}^J \frac{\partial}{\partial u_{J}^{\alpha}} \\*[3ex]
J\equiv J(k)=(j_1,\ldots,j_k),\qquad 1\leq j_k\leq p,\qquad k=j_1+\cdots +j_k.
\end{array}
\end{equation}
The prolongation is defined on the corresponding jet space $M^{(n)}\subset X\times U^{(n)}$. Here the coefficients $\phi_{\alpha}^{J}$ depend on $x,u$ and the derivatives of $u$ up to order $k$ for $J=J(k)$. These coefficients are obtained by the following formula \cite{Olver} 
\begin{equation}
\label{coeffphi}
\phi_{\alpha}^J=D_J\left(\phi_\alpha+\sum_{i=1}^p \xi^i u_i^\alpha\right)+\sum_{i=1}^p  \xi^i u_{J,i}^\alpha,
\end{equation}
where $D_J$ is the total derivative operator, $u_i^\alpha:=\frac{\partial u^\alpha}{\partial x_i}$ and $u_{J,i}^\alpha:=\frac{\partial u_J^\alpha}{\partial x_i}$.

Condition \eqref{algo} provides a system of linear ordinary differential equations of order $n$ for the coefficients $\xi^i$ and $\phi_\alpha$. The {\em determining equations} for the symmetry operator $\widehat v$ are obtained by setting equal to zero the coefficients of each linearly independent expressions in the derivative of the $u^\alpha$. Each vector field \eqref{vhat} provides a one parameter Lie group $\mathcal G(\lambda)$ obtained by integrating the system
\begin{equation}
\label{integratev}
\frac{d\widetilde x^i}{d\lambda}=\xi^i(\widetilde x,\widetilde u), \quad \widetilde x^i\Big|_{\lambda=0}=x^i, \quad \qquad 
\frac{d\widetilde u^\alpha}{d\lambda}=\phi^\alpha(\widetilde x,\widetilde u), \quad \widetilde u^\alpha\Big|_{\lambda=0}=u^\alpha.
\end{equation}


\subsection{Lie point symmetries of the quaternionic Riccati \\ equation}
We now apply the Lie symmetry tool, presented in the last subsection, for the case of the quaternionic Riccati equation \eqref{riccati3}. Decomposing explicitly equation \eqref{riccati3} in terms of its components, where for convenience we set $\mathbf Q:=ue_1+ve_2+we_3$ for $u,v,w,q\in C^1(\Omega,\R)$, the system \eqref{riccatidec} is then equivalent to
\begin{subequations}
\begin{eqnarray}
\label{E1}
E_1 &\equiv & -(\partial_1 u+\partial_2 v+\partial_3 w)+(u^2+v^2+w^2)-q =0, \\
\label{E2}
E_2 &\equiv &\partial_2 w-\partial_3 v =0, \\
\label{E3}
E_3 &\equiv & \partial_3 u-\partial_1 w =0, \\
\label{E4}
E_4 &\equiv &\partial_1 v-\partial_2 u =0. 
\end{eqnarray}
\end{subequations}
For this system of equations the vector fields \eqref{vhat} can be expressed as
\begin{equation}
\label{v}
\widehat v = \xi \partial_{x}+\eta\partial_{y}+ \tau \partial_{z} +\phi \partial_u+\psi \partial_v+ \zeta \partial_w, 
\end{equation}
where the real-valued coefficient functions $\xi,\ldots,\zeta$ depend on the three independent variables $x,y,z$ and the three dependent variables $u,v,w$. Considering now the algorithm given by \eqref{algo}, i.e. applying the vector fields \eqref{v} on the system of equations \eqref{E1},...,\eqref{E4} on the solution set of the Riccati equation \eqref{riccati3}, we find a system of 37 determining equations (the calculation was checked using {\em PDEtools} in Maple~16). We find that the generator $\widehat v$ must actually have the form
\begin{equation}
\label{vparameters}
\begin{array}{l}
\widehat v = \Big[a_1 \big(x^2- (y^2+z^2)\big) +2 a_2 x y + 2 a_3 x z - a_4 z + a_5 x - a_6 y + a_9\Big]\partial_x\\*[2ex]
+ \Big[2 a_1 x y + a_2 \big(y^2 - (x^2+z^2)\big) +2 a_3 y z + a_5 y + a_6 x - a_7 z + a_{10}\Big] \partial_y\\*[2ex]
+ \Big[2 a_1 x z + 2 a_2 y z + a_3 \big(z^2 - (x^2+y^2)\big) + a_4 x + a_5 z+a_7 y + a_8\Big] \partial_z\\*[2ex]
+ \Big[a_1 \big(1-2 (x u + y v + z w)\big) + 2a_2 (x v - yu) + 2 a_3 (x w - zu) \\*[2ex] - a_4 w - a_5 u - a_6 v \Big] \partial_u + \Big[2 a_1 (y u - x v) + a_2 \big(1-2(x u + y v + z w)\big) \\*[2ex] - 2 a_3 (z v - y w) - a_5 v + a_6 u - a_7 w \Big] \partial_v +\Big[2 a_1 (z u - x w)\\*[2ex] + 2 a_2 (z v - y w) + a_3 \big(1-2(x u + y v + z w)\big) + a_4 u - a_5 w + a_7 v\Big]\partial_w,
\end{array}
\end{equation}
for some real parameters $a_1, \ldots, a_{10}$. The potential $q$ being an arbitrary real-valued function is subject to one further determining equation that involve $q$ explicitly:
\begin{equation}
\label{condq}
\begin{array}{l}
\Big[ a_1 (y^2+z^2 - x^2) -2 a_2 x y - 2 a_3 x z + a_4 z - a_5 x + a_6 y - a_9\Big] q_x\\*[2ex]
+\Big[ -2 a_1 x y + a_2 (x^2 + z^2 - y^2) - 2 a_3 y z - a_5 y - a_6 x + a_7 z - a_{10}\Big] q_y\\*[2ex]
+\Big[ -2 a_1 x z - 2 a_2 y z + a_3 (x^2 + y^2 - z^2) - a_4 x - a_5 z - a_7 y - a_8\Big] q_z \\*[2ex]
- 2\Big[a_5 + 2(a_1 x + a_2 y + a_3 z)\Big] q = 0.
\end{array}
\end{equation}
For each non-zero parameters $a_k$, $k=1,\ldots, 10$, a Lie point symmetry of the Riccati equation \eqref{riccati3} exists if and only if a solution of \eqref{condq} is obtained. Hence, solving equation \eqref{condq} for each parameter $a_k$, i.e. obtaining the form of the potential function $q$, the associate generator $\widehat v$ is given from \eqref{vparameters}. For convenience in what follows we define $\mathbf x:=(x,y,z)$, $\mathbf Q(\mathbf x):=(u(\mathbf x),v(\mathbf x),w(\mathbf x))$, $r:=\sqrt{x^2+y^2+z^2}$ and the function $c(\mathbf x,\mathbf y):=1-2\langle \mathbf x,\mathbf y\rangle$. We present the results in the following table:

\begin{center}
\begin{tabular}{|l|l|}
\hline
{\bf Vector fields} & {\bf Potentials} $q(x,y,z)$ \\*[2ex]
\hline\hline
$\widehat v_1=\partial_x$ & $q=F_1(y,z)$  \\*[2ex]
$\widehat v_2=\partial_y$ & $q=F_2(x,z)$ \\*[2ex]
$\widehat v_3=\partial_z$ & $q=F_3(x,y)$ \\*[2ex]
$\widehat v_4=y\partial_z-z\partial_y+v\partial_w-w\partial_v$ & $q=F_4\big(x,\sqrt{y^2+z^2}\big)$ \\*[2ex]
$\widehat v_5=z\partial_x-x\partial_z+w\partial_u-u\partial_w$ & $q=F_5\big(y,\sqrt{x^2+z^2}\big)$ \\*[2ex]
$\widehat v_6=x\partial_y-y\partial_x+u\partial_v-v\partial_u$ & $q=F_6\big(z,\sqrt{x^2+y^2}\big)$ \\*[2ex]
$\widehat v_7=x\partial_x+y\partial_y+z\partial_z-u\partial_u-v\partial_v-w\partial_w$ & $q=x^{-2}F_7\Big(\displaystyle \frac{y}{x},\displaystyle \frac{z}{x}\Big)$ \\*[2ex]
$\widehat v_8=[x^2-(y^2+z^2)]\partial_x+2xy\partial_y+2xz\partial_z $ & $q=r^{-4}F_8\Big(\displaystyle \frac{y}{r^2},\displaystyle \frac{z}{r^2}\Big)$\\
\hspace*{6mm}$+c\big(\mathbf x,\mathbf Q(\mathbf x)\big)\partial_u+2(yu-xv)\partial_v+2(zu-xw)\partial_w$ & \\*[2ex]
$\widehat v_9=2xy\partial_x+[y^2-(x^2+z^2)]\partial_y+2yz\partial_z $ & $q=r^{-4}F_9\Big(\displaystyle \frac{x}{r^2},\displaystyle \frac{z}{r^2}\Big)$\\
\hspace*{6mm}$+2(xv-yu)\partial_u+c\big(\mathbf x,\mathbf Q(\mathbf x)\big)\partial_v+2(zv-yw)\partial_w$ & \\*[2ex]
$\widehat v_{10}=2xz\partial_x+2yz\partial_y+[z^2-(x^2+y^2)]\partial_z$ & $q=r^{-4}F_{10}\displaystyle \Big(\frac{x}{r^2},\displaystyle \frac{y}{r^2}\Big)$\\
\hspace*{6mm}$+2(xw-zu)\partial_u+2(yw-zv)\partial_v+c\big(\mathbf x,\mathbf Q(\mathbf x)\big)\partial_w$ & \\
\hline
\end{tabular}\\*[2ex]
{\bf Table 1}
\end{center}

Generators $\widehat v_1,\ldots, \widehat v_3$ represent translations in the space of independent variables. The generators $\widehat v_4,\ldots, \widehat v_6$ represent simultaneous rotations in the space of independent and dependent variables and $\widehat v_7$ represents dilations. Finally, the generators $\widehat v_8,\ldots, \widehat v_{10}$ represent conical symmetries.

In what follows, we define the function $\alpha(x,r,\lambda):=r^2\lambda^2-2x\lambda+1$ and $\mathbf x^\top$ represents the column vector of $\mathbf x$. In particular, we note that $\alpha(x,r,\lambda)$ is a quadratic polynomial in $\lambda$ without real roots except on the $x$-axis. The one-parameter groups $\mathcal G_k$ generated by the $\widehat v_k$ are obtained using \eqref{integratev}. 

$$
\begin{array}{lll}
\mathcal G_1: &\widetilde {\mathbf x}=\mathbf x + \lambda \mathbf e_1, & \widetilde{\mathbf Q}=\mathbf Q, \qquad \mathbf e_1:=(1,0,0), \\
\mathcal G_2: &  \widetilde{\mathbf x}=\mathbf x + \lambda \mathbf e_2, & \widetilde{\mathbf Q}=\mathbf Q,  \qquad \mathbf e_2:=(0,1,0),\\
\mathcal G_3: & \widetilde{\mathbf x}=\mathbf x + \lambda \mathbf e_3, & \widetilde{\mathbf Q}=\mathbf Q, \qquad \mathbf e_3:=(0,0,1), \\*[2ex]
 \mathcal G_4: &  \widetilde {\mathbf x}^\top= R_1(\lambda)\mathbf x^\top, & \widetilde{\mathbf Q}^\top= R_1(\lambda)\mathbf Q^\top, \quad  R_1(\lambda):=\left(\begin{array}{ccc} 1 & 0 & 0 \\
0 & \cos \lambda & \sin \lambda \\
0 & -\sin \lambda & \cos\lambda
\end{array}\right),\\*[4ex]
 \mathcal G_5 :&   \widetilde {\mathbf x}^\top= R_2(\lambda)\mathbf x^\top, & \widetilde{\mathbf Q}^\top= R_2(\lambda)\mathbf Q^\top, \quad R_2(\lambda):= \left(\begin{array}{ccc} \cos \lambda & 0 & -\sin \lambda \\
0 & 1 & 0 \\
\sin \lambda & 0 & \cos\lambda
\end{array}\right),\\*[4ex]
 \mathcal G_6 :&   \widetilde {\mathbf x}^\top= R_3(\lambda)\mathbf x^\top, & \widetilde{\mathbf Q}^\top= R_3(\lambda)\mathbf Q^\top, \quad R_3(\lambda):= \left(\begin{array}{ccc} \cos\lambda & \sin\lambda & 0 \\
-\sin\lambda & \cos \lambda & 0 \\
0 & 0 & 1
\end{array}\right),
\\*[2ex]
\mathcal G_7: &\widetilde {\mathbf x}=e^\lambda \mathbf x, & \mathbf{\widetilde Q}=e^{-\lambda}\mathbf Q,  \\*[2ex]
\mathcal G_8: & \widetilde{\mathbf x}=\displaystyle\frac{\mathbf x-\lambda r^2\mathbf e_1}{\alpha(x,r,\lambda)}, & \widetilde{\mathbf Q}^\top=\left(\begin{array}{l} u+c(\mathbf x,\mathbf Q)\lambda-\big(r^2u+c(\mathbf x,\mathbf Q)x\big)\lambda^2 \\
\alpha(x,r,\lambda)v+ y\big(2u\lambda+c(\mathbf x,\mathbf Q)\lambda^2\big) \\ \alpha(x,r,\lambda)w+ z\big(2u\lambda+c(\mathbf x,\mathbf Q)\lambda^2\big)
\end{array}\right),\ y^2+z^2\neq 0, \\*[4ex]
&   \widetilde{\mathbf x}= \big(\displaystyle \frac{x}{1-x\lambda},0,0\big), & \widetilde{\mathbf Q}^\top=\left(\begin{array}{l} 
(1-x\lambda)\big[u+\lambda(1-xu)] \\
v(1-x\lambda)^2 \\
w(1-x\lambda)^2 
\end{array}\right),\ y^2+z^2= 0, \\*[4ex]

\mathcal G_9: & \widetilde{\mathbf x}=\displaystyle\frac{\mathbf x-\lambda r^2\mathbf e_2}{\alpha(y,r,\lambda)}, & \widetilde{\mathbf Q}^\top=\left(\begin{array}{l} \alpha(y,r,\lambda)u+ x\big(2v\lambda+c(\mathbf x,\mathbf Q)\lambda^2\big) \\ v+c(\mathbf x,\mathbf Q)\lambda-\big(r^2v+c(\mathbf x,\mathbf Q)y\big)\lambda^2 \\ \alpha(y,r,\lambda)w+ z\big(2v\lambda+c(\mathbf x,\mathbf Q)\lambda^2\big) \\ \end{array}\right),\ x^2+z^2\neq 0, \\*[4ex]
&   \widetilde{\mathbf x}= \big(0,\displaystyle \frac{y}{1-y\lambda},0\big), & \widetilde{\mathbf Q}^\top=\left(\begin{array}{l} 
u(1-y\lambda)^2 \\
(1-y\lambda)\big[v+\lambda(1-yv)]\\
w(1-y\lambda)^2 
\end{array}\right),\ x^2+z^2= 0, \\*[4ex]

\mathcal G_{10}: & \widetilde{\mathbf x}=\displaystyle\frac{\mathbf x-\lambda r^2\mathbf e_3}{\alpha(z,r,\lambda)}, & \widetilde{\mathbf Q}^\top=\left(\begin{array}{l} \alpha(z,r,\lambda)u+ x\big(2w\lambda+c(\mathbf x,\mathbf Q)\lambda^2\big) \\ \alpha(z,r,\lambda)v+ y\big(2w\lambda+c(\mathbf x,\mathbf Q)\lambda^2\big) \\ w+c(\mathbf x,\mathbf Q)\lambda-\big(r^2w+c(\mathbf x,\mathbf Q)z\big)\lambda^2\end{array}\right),\ x^2+y^2\neq 0, \\*[4ex]
&  \widetilde{\mathbf x}= \big(0,0,\displaystyle \frac{z}{1-z\lambda}\big), & \widetilde{\mathbf Q}^\top=\left(\begin{array}{l} 
u(1-z\lambda)^2 \\
v(1-z\lambda)^2 \\
(1-z\lambda)\big[w+\lambda(1-zw)]
\end{array}\right),\ x^2+y^2= 0.
\end{array}
$$

Since each group $\mathcal G_k$ is a symmetry group associated with a potential $q$ written in term of $F_k$, if $\mathbf Q(\mathbf x)=\Big(f(\mathbf x),\ g(\mathbf x),\ h(\mathbf x)\Big)$ is a solution of \eqref{riccati3} so are the functions


$$
\begin{array}{rcl}
\mathbf Q^{(1)}&=& \Big(f(\mathbf x-\lambda \mathbf e_1),\ g(\mathbf x-\lambda \mathbf e_1),\ h(\mathbf x-\lambda \mathbf e_1)\Big),\\
\mathbf Q^{(2)}&=& \Big(f(\mathbf x-\lambda \mathbf e_2),\ g(\mathbf x-\lambda \mathbf e_2),\ h(\mathbf x-\lambda \mathbf e_2)\Big),\\
\mathbf Q^{(3)}&=& \Big(f(\mathbf x-\lambda \mathbf e_3),\ g(\mathbf x-\lambda \mathbf e_3),\ h(\mathbf x-\lambda \mathbf e_3)\Big), \\*[2ex]
\mathbf Q^{(4)}&=& R_1(\lambda)\Big(f(R_1^{\top}\mathbf x),\ g(R_1^{\top}\mathbf x),\ h(R_1^{\top}\mathbf x)\Big)^\top,\\*[2ex]
\mathbf Q^{(5)}&=& R_2(\lambda)\Big(f(R_2^{\top}\mathbf x),\ g(R_2^{\top}\mathbf x),\ h(R_2^{\top}\mathbf x)\Big)^\top,
 \\*[2ex]
\mathbf Q^{(6)}&=& R_3(\lambda)\Big(f(R_3^{\top}\mathbf x),\ g(R_3^{\top}\mathbf x),\ h(R_3^{\top}\mathbf x)\Big)^\top,
 \\*[2ex]
\mathbf Q^{(7)}&=&e^{-\lambda}\Big(f(e^{-\lambda}\mathbf x),\  g(e^{-\lambda}\mathbf x),\ h(e^{-\lambda}\mathbf x)\Big), \\*[2ex]
\mathbf Q^{(8a)}&=& \Big(f(\frac{\mathbf x+r^2\lambda \mathbf e_1}{\alpha(-x,r,\lambda)})+c\big(\mathbf x,\mathbf Q(\frac{\mathbf x+r^2\lambda \mathbf e_1}{\alpha(-x,r,\lambda)})\big)(\lambda-x\lambda^2)-r^2\lambda^2 f\big(\frac{\mathbf x+r^2\lambda \mathbf e_1}{\alpha(-x,r,\lambda)}\big),\\*[2ex]
&& \alpha(x,r,\lambda)g(\frac{\mathbf x+r^2\lambda \mathbf e_1}{\alpha(-x,r,\lambda)})+y\Big[2f(\frac{\mathbf x+r^2\lambda \mathbf e_1}{\alpha(-z,r,\lambda)})\lambda+ c\big(\mathbf x,\mathbf Q(\frac{\mathbf x+r^2\lambda \mathbf e_1}{\alpha(-x,r,\lambda)})\big)\lambda^2\Big],\\*[2ex]
&& \alpha(x,r,\lambda)h(\frac{\mathbf x+r^2\lambda \mathbf e_1}{\alpha(-x,r,\lambda)})+z\Big[2f(\frac{\mathbf x+r^2\lambda \mathbf e_1}{\alpha(-x,r,\lambda)})\lambda+ c\big(\mathbf x,\mathbf Q(\frac{\mathbf x+r^2\lambda \mathbf e_1}{\alpha(-x,r,\lambda)})\big)\lambda^2\Big]\Big),\quad \text{for } y^2+z^2\neq 0,\\*[2ex]
\mathbf Q^{(8b)} &=& \Big((1-x\lambda)^2f(0,0,{\frac{x}{1+x\lambda}})+(1-x\lambda)\lambda,\ (1-x\lambda)^2g(0,0,{\frac{x}{1+x\lambda}}),\\*[2ex] && (1-x\lambda)^2 h(0,0,{\frac{x}{1+x\lambda}})\Big),\quad \text{for }y^2+z^2=0\\*[2ex]

\mathbf Q^{(9a)}&=& \Big( \alpha(y,r,\lambda)f(\frac{\mathbf x+r^2\lambda \mathbf e_2}{\alpha(-y,r,\lambda)})+x\Big[2g(\frac{\mathbf x+r^2\lambda \mathbf e_2}{\alpha(-y,r,\lambda)})\lambda+ c\big(\mathbf x,\mathbf Q(\frac{\mathbf x+r^2\lambda \mathbf e_2}{\alpha(-y,r,\lambda)})\big)\lambda^2\Big],\\*[2ex] 

&& g(\frac{\mathbf x+r^2\lambda \mathbf e_2}{\alpha(-y,r,\lambda)})+c\big(\mathbf x,\mathbf Q(\frac{\mathbf x+r^2\lambda \mathbf e_2}{\alpha(-y,r,\lambda)})\big)(\lambda-y\lambda^2)-r^2\lambda^2 g\big(\frac{\mathbf x+r^2\lambda \mathbf e_2}{\alpha(-y,r,\lambda)}\big), \\*[2ex]

&& \alpha(y,r,\lambda)h(\frac{\mathbf x+r^2\lambda \mathbf e_2}{\alpha(-y,r,\lambda)})+z\Big[2g(\frac{\mathbf x+r^2\lambda \mathbf e_2}{\alpha(-y,r,\lambda)})\lambda+ c\big(\mathbf x,\mathbf Q(\frac{\mathbf x+r^2\lambda \mathbf e_2}{\alpha(-y,r,\lambda)})\big)\lambda^2\Big]\Big),\ \text{for } x^2+z^2\neq 0,\\*[2ex]

\mathbf Q^{(9b)} &=& \Big((1-y\lambda)^2 f(0,0,{\frac{y}{1+y\lambda}}),\ (1-y\lambda)^2 g(0,0,{\frac{y}{1+y\lambda}})+(1-y\lambda)\lambda,\\*[2ex]
&& (1-y\lambda)^2 h(0,0,{\frac{y}{1+y\lambda}})\Big), \quad \text{for }x^2+z^2=0\\*[2ex]

\mathbf Q^{(10a)}&=& \Big(\alpha(z,r,\lambda)f(\frac{\mathbf x+r^2\lambda \mathbf e_3}{\alpha(-z,r,\lambda)})+x\Big[2h(\frac{\mathbf x+r^2\lambda \mathbf e_3}{\alpha(-z,r,\lambda)})\lambda+ c\big(\mathbf x,\mathbf Q(\frac{\mathbf x+r^2\lambda \mathbf e_3}{\alpha(-z,r,\lambda)})\big)\lambda^2\Big],\\*[2ex]
&& \alpha(z,r,\lambda)g(\frac{\mathbf x+r^2\lambda \mathbf e_3}{\alpha(-z,r,\lambda)})+y\Big[2h(\frac{\mathbf x+r^2\lambda \mathbf e_3}{\alpha(-z,r,\lambda)})\lambda+ c\big(\mathbf x,\mathbf Q(\frac{\mathbf x+r^2\lambda \mathbf e_3}{\alpha(-z,r,\lambda)})\big)\lambda^2\Big],\\*[2ex]
&& h(\frac{\mathbf x+r^2\lambda \mathbf e_3}{\alpha(-z,r,\lambda)})+c\big(\mathbf x,\mathbf Q(\frac{\mathbf x+r^2\lambda \mathbf e_3}{\alpha(-z,r,\lambda)})\big)(\lambda-z\lambda^2)-r^2\lambda^2 h\big(\frac{\mathbf x+r^2\lambda \mathbf e_3}{\alpha(-z,r,\lambda)}\big)\Big),\quad \text{for } x^2+y^2\neq 0,\\*[2ex]
\mathbf Q^{(10b)} &=& \Big((1-z\lambda)^2f(0,0,{\frac{z}{1+z\lambda}}),\ (1-z\lambda)^2g(0,0,{\frac{z}{1+z\lambda}}),\\*[2ex]
&& (1-z\lambda)^2h(0,0,{\frac{z}{1+z\lambda}})+(1-z\lambda)\lambda\Big), \quad \text{for }x^2+y^2=0,
\end{array}
$$
where each new solution $\mathbf Q^{(k)}$ is associated with a potential function $q$ associated with the generator $\widehat v_k$ in Table~1.

\section{Applications using symmetry reductions}
Another application of symmetry methods is to reduce systems of differential equations, finding equivalent system of simpler form. This application is called {\em reduction}. In this section we perform some symmetry reductions of the quaternionic Riccati equation \eqref{riccati3} using symmetry groups obtained in the last section. We consider solutions of these symmetry reductions and combine them with the results found in section~2. We refere to \cite{Olver} chapter 3 for symmetry reduction using Lie group symmetries.

\subsection{Rotationally-invariant solutions}
For the group of rotations about de $z$-axis and the $w$-axis generated by $\widehat v_6=x\partial_y-y\partial_x+u\partial_v-v\partial_u$, invariants are provided by $\rho=\sqrt{x^2+y^2}$, $z$, $\hat{u}=u\cos{\theta} +v\sin{\theta}$, $\hat{v}=-u\sin{\theta} +v\cos{\theta}$ and $\hat{w}=w$. The reduced equations $\eqref{E1},\ldots,\eqref{E4}$ are then
\begin{subequations}
\begin{align}
\label{R1}
E_1 &\equiv -\left(\hat{u}_\rho +\frac{\hat{u}}{\rho}+\hat{w}_z\right)+\hat{u}^2+\hat{v}^2+\hat{w}^2-F_6(z,\rho)=0, \\
\label{R2}
E_2 &\equiv \hat{v}_z+\left(\hat{u}_z-\hat{w}_\rho\right)\tan{\theta}=0, \\
\label{R3}
E_3 &\equiv \hat{u}_z-\hat{w}_\rho -\hat{v}_z\tan{\theta}=0, \\
\label{R4}
E_4 &\equiv \hat{v}_\rho  +\frac{\hat{v}}{\rho}=0,
\end{align}
\end{subequations} 
where $F_6$ is the potential for the generator $\widehat v_6$ in Table 1.

The system of equation $\eqref{R2}$ and $\eqref{R3}$ can be simplified as
\begin{equation}
\label{R5}
E_2' \equiv \hat{u}_z-\hat{w}_\rho=0 
\end{equation}
and
\begin{equation}
\label{R6}
E_3' \equiv \hat{v}_z=0.
\end{equation}
 From $\eqref{R6}$ and $\eqref{R4}$ we find $\hat{v}=\frac{c_1}{\rho}$, where $c_1$ is a real constant. If we further assume translational invariance under $\widehat v_3=\partial_z$, equation (\ref{R5}) implies that $\hat w=c_2$, where $c_2$ is a real constant. We have now to solve equation \eqref{R1} which can be rewritten as

\begin{equation}
\label{RRR}
\hat{u}_\rho = \hat{u}^2-\displaystyle \frac{\hat{u}}{\rho}-\widetilde F_6(\rho),
\end{equation}
where $\widetilde F_6(\rho):=F_6(\rho)-\left(\frac{c_1}{\rho}\right)^2-c_2^2$. Since $F_6$ is an arbitrary function of $\rho$, the constants $c_1, c_2$ can be absorbed such that $c_1=c_2=0$. Equation \eqref{RRR} is a one-dimensional (classical) Riccati equation of the form \eqref{riccati1}. Therefore, following \eqref{linRiccati} equation \eqref{RRR} is equivalent to
\begin{equation}
\label{linriccati2}
g_{\rho\rho}+\frac{g_\rho}{\rho}-\widetilde F_6\,g=0,
\end{equation}
where solution $g(\rho)$ of \eqref{linriccati2} leads to solution $\hat u=-g_{\rho}/g$ of \eqref{RRR}.

Assuming now a specific form of the potential, we consider
\begin{equation}
\widetilde F_6(\rho)=\displaystyle\frac{k^2}{\rho^2},
\end{equation}
where $k$ is a real constant. Solution of equation \eqref{RRR} is given by 

\begin{equation}
\label{SRR}
\hat{u}(\rho)=-\frac{k}{\rho}\frac{\left(\rho^{2k}+e^{2ck}\right)}{\left(\rho^{2k}-e^{2ck}\right)}
\end{equation}
where $c$ is a real constant. Since we have found $\hat u(\rho)$ and that $\hat v(\rho)=0,\ \hat w(\rho)=0$, we can now obtain the inverse transformations
\begin{equation}
\begin{array}{rcl}
u&=&\displaystyle\frac{\hat{u}x}{\sqrt{x^2+y^2}}-\displaystyle\frac{\hat{v}y}{\sqrt{x^2+y^2}},\\*[2ex]
v&=&\displaystyle\frac{\hat{u}y}{\sqrt{x^2+y^2}}+\displaystyle\frac{\hat{v}x}{\sqrt{x^2+y^2}},\\*[2ex]
w&=&\hat{w},
\end{array}
\end{equation}
such that
\begin{equation}
\label{Solution Riccati}
\begin{array}{l}
u=\displaystyle \frac{-kx\left[\left(x^2+y^2\right)^k+e^{2ck} \right]}{\left(x^2+y^2\right)\left[\left(x^2+y^2\right)^k-e^{2ck}\right]},\\*[2ex]
v=\displaystyle\frac{-ky\left[\left(x^2+y^2\right)^k+e^{2ck}\right]}{\left(x^2+y^2\right)\left[\left(x^2+y^2\right)^k-e^{2ck}\right]},\\*[2ex]
w=0.
\end{array}
\end{equation}
Thus for a potential of the form $q=\frac{k^2}{x^2+y^2}$, a solution of the quaternionic Riccati equation  $\eqref{riccati3}$ is given by $\mathbf{Q}=ue_1+ve_2+we_3$.

Now using theorem \ref{expA} a solution $\psi$ of the Schr\"odinger equation $\eqref{schrodinger3}$ can be found using this solution $\mathbf Q$ of the three-dimensional Riccati equation:
\begin{equation}
\begin{array}{rcl}
\psi(x,y,z)&=&\exp{\left(-\mathcal A[\mathbf{Q}]\right)}\\*[2ex]
&=&\exp{\left(- \left(\displaystyle \int_{1}^x{u(\xi,0,0)d\xi}+\displaystyle\int_{0}^y{v(x,\eta,0)d\eta}+\displaystyle\int_{0}^z{w(x,y,\zeta)d\zeta}-\ln{C}\right)\right)}\\*[2ex]
&=&\displaystyle\frac{C\left(\left(x^2+y^2\right)^ke^{-2ck}-1\right)}{\left(x^2+y^2\right)^\frac{k}{2}\left(e^{-2ck}-1\right)}.
\end{array}
\end{equation}
Therefore, a solution of the Schr\"odinger equation $\eqref{schrodinger3}$ is given by
\begin{equation}
\label{Solution Schrodinger}
\psi(x,y,z)=\frac{C\left(\left(x^2+y^2\right)^k-e^{2ck}\right)}{\left(x^2+y^2\right)^\frac{k}{2}\left(1-e^{2ck}\right)}
\end{equation} 
for a potential $q=\frac{k^2}{x^2+y^2}$.

\subsection{Conically-invariant solutions} 
Let us now look at the one-parameter group $\mathcal G_{10}$ of conical symmetry generated  by 
\begin{equation}
\begin{array}{rcl}
\widehat v_{10}&=&2xz\partial_x+2yz\partial_y+[z^2-(x^2+y^2)]\partial_z\\*[2ex]
&+& 2(xw-zu)\partial_u+2(yw-zv)\partial_v+c\big(\mathbf x,\mathbf Q(\mathbf x)\big)\partial_w
\end{array}
\end{equation}
for the potential $q=r^{-4}F_{10}\displaystyle \Big(\frac{x}{r^2},\displaystyle \frac{y}{r^2}\Big)$ and $x^2+y^2\neq 0$.

Independent invariants can be easily calculated, we obtain
\begin{equation}
s=\frac{x}{r^2}\qquad \text{ and } \qquad t=\frac{y}{r^2}.
\end{equation}

We have now to solve the following system
\begin{equation}
\label{syscon}
dx=\frac{xzdu}{xw-zu}=\frac{xzdv}{yw-zv}=\frac{xzdw}{\frac{1}{2}-xu-yv-zw},
\end{equation}
where an immediate solution of $v$ is given by $v=\frac{t}{s}u$. The system \eqref{syscon} then becomes

\begin{equation}
dr=\frac{du}{\frac{2sw}{\sqrt{1-r^2(s^2+t^2)}}-\frac{2u}{r}}=\frac{dw}{\frac{1}{r^2\sqrt{1-r^2(s^2+t^2)}}-\frac{2(s^2+\frac{t^2}{s})u}{\sqrt{1-r^2(s^2+t^2)}}-\frac{2w}{r}}.
\end{equation}
We find the following solutions
\begin{equation}
\begin{array}{rcl}
u&=&-\frac{1-2r^2(s^2+t^2)}{r^2}U+\frac{\sqrt{1-r^2(s^2+t^2)}}{r}iV+s \\*[2ex]
w&=&\frac{2(s^2+t^2)\sqrt{1-r^2(s^2+t^2)}}{sr}U+\frac{1-2r^2(s^2+t^2)}{2r^2s}iV+\frac{\sqrt{1-r^2(s^2+t^2)}}{r}.
\end{array}
\end{equation}
The dependent invariants are therefore given by

\begin{equation}
\begin{array}{rcl}
U&=&-r^2\big(1-2r^2(s^2+t^2)\big)u+2r^3s\sqrt{1-r^2(s^2+t^2)}w-r^2s\\*[2ex]
&=&(x^2+y^2-z^2)u+(2xz)w-x, \\*[2ex]
V&=&i\left[-4r^3(s^2+t^2)\sqrt{1-r^2(s^2+t^2)}u-2r^2s\big(1-2r^2(s^2+t^2)\big)w+2rs\sqrt{1-r^2(s^2+t^2)}\right]\\*[2ex]
&=& i\left[\frac{-4z(x^2+y^2)}{r^2}u+\frac{2x}{r^2}\left(x^2+y^2-z^2\right)w+\frac{2xz}{r^2}\right].
\end{array}
\end{equation}

Calculating now the terms $u_x,\ldots , w_z$ and substituing into the system of equations  $\eqref{E1},\ldots,\eqref{E4}$, we find

\begin{flalign*} 
E_1 \equiv& -\frac{1}{4s}V^2+\left(\frac{s^2+t^2}{s}\right)U^2+U+tU_t+sU_s=sF_{10}(s,t), &&\\*[2ex]
E_2 \equiv& -\left[2r^2t\right]iV + \left[4rt\sqrt{1-r^2(s^2+t^2)}\right]U \\
&+\left[1-2r^2s^2\right]iV_t+\left[2r^2st\right]iV_s\\
&+\left[4rs^2\sqrt{1-r^2(s^2+t^2)}\right]U_t -\left[4rst\sqrt{1-r^2(s^2+t^2)}\right]U_s=0,\\*[2ex]
E_3 \equiv& \left[1-2r^2t^2\right]iV + \left[4rt^2\sqrt{1-r^2(s^2+t^2)}\right]U \\
&-\left[2r^2s^2t\right]iV_t+\left[s(2r^2t^2-1)\right]iV_s\\
&+\left[4rs^2t\sqrt{1-r^2(s^2+t^2)}\right]U_t -\left[4rst^2\sqrt{1-r^2(s^2+t^2)}\right]U_s=0,\\*[2ex]
E_4 \equiv& -\left[rt\sqrt{1-r^2(s^2+t^2)}\right]iV + \left[t\big(1-2r^2(s^2+t^2)\big)\right]U \\
&-\left[rs^2\sqrt{1-r^2(s^2+t^2)}\right]iV_t+\left[rst\sqrt{1-r^2(s^2+t^2)}\right]iV_s\\
&+\left[s^2\big(1-2r^2(s^2+t^2)\big)\right]U_t -\left[st\big(1-2r^2(s^2+t^2)\big)\right]U_s=0.
\end{flalign*}
Considering the new equations $E_2'=tE_2-E_3$, $E_3'=tE_2+E_3$, as well as  $E_3''=\sqrt{1-r^2(s^2+t^2)}E_3'-4rtE_4$, $E_4'=\sqrt{1-r^2(s^2+t^2)}E_3'+4rtE_4$, we obtain $V=sV_s+tV_t$. We find  $V_t=0$ and the simplified system of equations $\eqref{E1},\ldots,\eqref{E4}$ becomes

\begin{equation} 
\frac{s^2+t^2}{s^2}U^2+\frac{1}{s}U+\frac{t}{s}U_t+U_s=F_{10}-\left(\frac{C_1}{2}\right)^2, \qquad \frac{t}{s}U+sU_t -tU_s=0.
\end{equation}
with
\begin{equation} 
V(s)=iC_1s,\qquad C_1\text{ a real constant.}
\end{equation}

To go further in this symmetry reduction let us assume that $F_{10}=(\frac{C_1}{2})^2$, i.e. the potential is of the form $q=\left(\displaystyle \frac{C_1}{2r^2}\right)^2$. Then we obtain
\begin{equation}
U(s,t)=\frac{2s}{(s^2+t^2)\ln{\big(C_2(s^2+t^2)\big)}},
\end{equation}
where $C_2$ is a real constant, such that
\begin{equation}
\begin{array}{rcl}
u&=&-\displaystyle \frac{C_1xz}{r^4}+\displaystyle \frac{2x(x^2+y^2-z^2)}{r^2(x^2+y^2)\ln{\left(C_2\left(\displaystyle \frac{x^2+y^2}{r^4}\right)\right)}}+\displaystyle \frac{x}{r^2},\\*[2ex]
v&=&-\displaystyle \frac{C_1yz}{r^4}+\displaystyle \frac{2y(x^2+y^2-z^2)}{r^2(x^2+y^2)\ln{\left(C_2\left(\displaystyle \frac{x^2+y^2}{r^4}\right)\right)}}+\displaystyle \frac{y}{r^2},\\*[2ex]
w&=&\displaystyle \frac{C_1(x^2+y^2-z^2)}{2r^4}+\displaystyle \frac{4z}{r^2\ln{\left(C_2\left(\displaystyle \frac{x^2+y^2}{r^4}\right)\right)}}+\displaystyle \frac{z}{r^2}.
\end{array}
\end{equation}
Therefore, $\mathbf Q=ue_1+ve_2+we_3$ is a quaternionic solution of the three-dimensional Riccati equation with potential $q(x,y,z)=\left(\displaystyle \frac{C_1}{2r^2}\right)^2$.

Again here the solution $\mathbf Q$ can be used to obtain a solution $\psi$ of the Schr\"odinger equation $\eqref{schrodinger3}$ using theorem  \ref{expA}. We find
\begin{equation}
\begin{array}{rcl}
\psi(x,y,z)&=&\exp{\left(-\mathcal A[\mathbf{Q}]\right)}\\*[2ex]
&=&\exp{\left(- \left(\displaystyle \int_{1}^x{u(\xi,0,0)d\xi}+\displaystyle \int_{0}^y{v(x,\eta,0)d\eta}+\displaystyle \int_{0}^z{w(x,y,\zeta)d\zeta}-\ln{C}\right)\right)}\\*[2ex]
&=& \displaystyle \frac{C\ln{\left(C_2\frac{x^2+y^2}{\left(x^2+y^2+z^2\right)^2}\right)}}{\ln{C_2}\sqrt{x^2+y^2+z^2}\exp{\left(\frac{C_1z}{2\left(x^2+y^2+z^2\right)}\right)}}
\end{array}
\end{equation}
i.e. $\psi\in \ker \big(-\Delta+q(x,y,z)\big)$ for $q(x,y,z)=\left(\displaystyle \frac{C_1}{2(x^2+y^2+z^2)}\right)^2$.

\subsubsection*{Acknowledgement}

CP and ST would like to thank Benoit Huard from Northumbria University for his help in the use of softwares that were used to calculate Lie symmetries. CP acknowledges a scholarship from the Institut des Sciences Math\'ematiques.

\end{document}